\documentclass[11pt]{article}

\usepackage{amsmath,amssymb,graphicx,amsthm,dsfont}
\usepackage{color,soul}
\usepackage{xspace}
\usepackage{enumerate}
\usepackage{hyperref}
\hypersetup{
citecolor=green!60!black,linkcolor=red!60!black,pagebackref=true,colorlinks
}
\usepackage{mathtools}

\usepackage{bbm}
\usepackage{etex}
\usepackage{comment}

\usepackage{nicefrac}
\usepackage{paralist}
\usepackage{todonotes}
\usepackage{setspace}
\usepackage[matrix,arrow,ps]{xy}

\usepackage{subcaption}
\usepackage[font=small,labelfont=bf]{caption}
\usepackage[noend,boxed,vlined,ruled,linesnumbered]{algorithm2e} 
\usepackage{xifthen}
\usepackage{booktabs}
\usepackage{tabularx}
\usepackage[shortlabels]{enumitem}

\usepackage[numbers]{natbib}

\usepackage{tikz}
\usetikzlibrary{decorations,arrows,petri,topaths,backgrounds,shapes,positioning,fit,calc,decorations.pathreplacing,patterns,intersections,decorations.pathmorphing,matrix}

\makeatletter
\newcommand{\gettikzxy}[3]{%
  \tikz@scan@one@point\pgfutil@firstofone#1\relax
  \edef#2{\the\pgf@x}%
  \edef#3{\the\pgf@y}%
}
\makeatother
\renewenvironment{description}
  {\list{}{\labelwidth=10pt \leftmargin=15pt
   }}
  {\endlist}

\usepackage[sort&compress,nameinlink]{cleveref}
\usepackage{etoolbox}

\newtheorem{theorem}{Theorem}
\newtheorem{lemma}{Lemma}
\newtheorem{corollary}{Corollary}
\newtheorem{proposition}{Proposition}

\theoremstyle{definition}
\newtheorem{observation}{Observation}

\theoremstyle{remark}
\newtheorem{example}{Example}

\theoremstyle{plain}

\crefname{table}{Table}{Tables}
\crefname{figure}{Figure}{Figures}
\crefname{theorem}{Theorem}{Theorems}
\Crefname{theorem}{Thm.}{Thms.}
\crefname{definition}{Definition}{Definitions}
\crefname{corollary}{Corollary}{Corollaries}
\crefname{observation}{Observation}{Observations}
\crefname{lemma}{Lemma}{Lemmas}
\crefname{example}{Example}{Examples}
\crefname{reduction}{Reduction}{Reductions}
\crefname{construction}{Construction}{Constructions}
\crefname{subsection}{Subsection}{Subsections}
\crefname{section}{Section}{Sections}
\crefname{proposition}{Proposition}{Propositions}
\Crefname{proposition}{Prop.}{Props.}
\crefname{algorithm}{Algorithm}{Algorithms}

\newcommand{\naturals}{{{\mathbb{N}}}}
\newcommand{\integers}{{{\mathbb{Z}}}}

\newcommand{\ib}{{\mathrm{IB}}}
\newcommand{\fair}{{\mathrm{Fair}}}
\renewcommand{\div}{{\mathrm{Div}}}
\renewcommand{\top}{{\mathrm{top}}}

\newcommand{\wone}{{\mathrm{W[1]}}}

\newcommand{\wtwo}{{\mathrm{W[2]}}}
\newcommand{\xp}{{\mathrm{XP}}}
\newcommand{\np}{{\mathrm{NP}}}
\newcommand{\fpt}{{\mathrm{FPT}}}

\DeclareMathOperator{\poly}{poly}

\renewcommand{\cal}[1]{\mathcal{#1}}
\newcommand{\calF}{\cal{F}}

\newcommand{\prob}[1]{\textsc{#1}}
\newcommand{\ibknap}{\prob{Individually Best Knapsack}}
\newcommand{\dknap}{\prob{Diverse Knapsack}}
\newcommand{\fknap}{\prob{Fair Knapsack}}

\newcommand{\Optproblemdef}[3]{
	\begin{center}
  \begin{minipage}{0.95\textwidth}
    \noindent
    \textsc{#1}\vspace{1pt}\\
    \setlength{\tabcolsep}{1pt}
    \renewcommand{\arraystretch}{1.1}
    \begin{tabularx}{\textwidth}{@{}lX@{}}
	    \textbf{Input:} 		& #2 \\
	    \textbf{Task:} 	& #3
    \end{tabularx}
  \end{minipage}
	\end{center}
}

\usepackage{authblk}
\newcommand{\mytitle}{Fair Knapsack}
\title{\mytitle\footnote{This research was initiated within the student project ``Research in Teams'' organized by the research group \emph{Algorithmics and Computational Complexity} of TU Berlin, Berlin, Germany.}
}

\author[1]{Till~Fluschnik\thanks{Supported by the DFG, projects DAMM (NI 369/13) and TORE (NI 369/18).}}
\author[2]{Piotr Skowron\thanks{Supported by a postdoctoral fellowship of the Alexander von Humboldt Foundation, Bonn, Germany and by the Foundation for Polish Science within the Homing programme (Project title: "Normative Comparison of Multiwinner Election Rules).}}
\author[1]{Mervin Triphaus} 
\author[1]{Kai Wilker}

\affil[1]{Algorithmics and Computational Complexity, Faculty~IV, TU Berlin, Germany\\
 \texttt{till.fluschnik@tu-berlin.de}, \texttt{p.skowron@mimuw.edu.pl}, \texttt{\{mervin.triphaus,wilker\}@campus.tu-berlin.de}
}
\affil[2]{University of Warsaw, Warsaw, Poland\\
 \texttt{p.skowron@mimuw.edu.pl}
}

\date{\vspace{-\baselineskip}}
 
\begin{document}
\sloppy
\allowdisplaybreaks

\maketitle

\begin{abstract}  %
We study the following multiagent variant of the knapsack problem. We are given a set of items, a set of voters, and a value of the budget; each item is endowed with a cost and each voter assigns to each item a certain value. The goal is to select a subset of items with the total cost not exceeding the budget, in a way that is consistent with the voters' preferences. Since the preferences of the voters over the items can vary significantly, we need a way of aggregating these preferences, in order to select the socially best valid knapsack. We study three approaches to aggregating voters' preferences, which are motivated by the literature on multiwinner elections and fair allocation. This way we introduce the concepts of individually best, diverse, and fair knapsack. 
We study the computational complexity (including parameterized complexity, and complexity under restricted domains) of the aforementioned multiagent variants of knapsack.          
\end{abstract}

\section{Introduction}\label{sec:intro}

In the classic knapsack problem we are given a set of \emph{items}, each having a cost and a value, and a budget. The goal is to find a subset of items 
with the maximal sum of the values subject to the constraint that the total cost of the selected items must not exceed the budget.
In this paper we are studying the following variant of the knapsack problem: instead of having a single objective value for each item we assume that
there is a set of \emph{agents} (also referred to as \emph{voters}) who have potentially different valuations (expressed through \emph{utilities}) of the items. When choosing a subset of items
we want to take into account possibly conflicting preferences of the voters with respect to which items should be selected:
in this paper we discuss three different approaches to how the voters' valuations can be aggregated. 

Multiagent knapsack forms an abstract model for a number of real-life scenarios. First, let us note that if the costs of the items are all the same, then the multiagent knapsack model collapses to the model for multiwinner elections~\cite{FSST-trends} (in the literature on multiwinner elections, items are often called candidates). 
Multiwinner voting rules are applicable in a broad class of scenarios, ranging from selecting a representative committee of experts, through recommendation systems~\cite{bou-lu:c:chamberlin-courant}\footnote{An example often described in the literature is when an enterprise considers which set of products should be pushed to production---it is natural to view such a problem as an instance of multiwinner elections with products corresponding to the items/candidates and potential customers to the voters.}, to resource allocation and facility location problems. In each of these settings it is quite natural to consider that different items/candidates can incur different costs. 
Further, algorithms for multiagent knapsack can be viewed as tools for the participatory budgeting problem~\cite{participatoryBudgeting}, where the authorities aggregate citizens' preferences in order to decide which of the potential local projects should obtain funding.

Perhaps the most straightforward way to aggregate voters' preferences is to select a subset (a knapsack) that maximizes the sum of the utilities of all the voters over all selected items. This approach---which we call selecting an \emph{individually best knapsack}---subject to differences in methods used for eliciting voters' preferences has been taken by Benabbou and Perny~\cite{BenPer16}, and in the context of participatory budgeting by Goel~et~al.~\cite{knapsackVoting} and Benade~et~al.~\cite{BNPS17}. 
However, by selecting an individually best knapsack we can disadvantage even large minorities of voters, which is illustrated by the following example: assume that the set of items can be divided into two subsets~$A_1$ and $A_2$, that all items have the same cost, and that $51\%$ of the voters like items from $A_1$ (assigning the utility of 1 to them, and the utility of 0 to the other items) and the remaining $49\%$ of voters like only items from $A_2$. An individually best knapsack would contain only items from $A_1$, that is $49\%$ of the voters would be effectively~disregarded.

In this paper we introduce two other approaches to aggregating voters' preferences for selecting a multiagent knapsack. 
One such approach---which we call selecting a \emph{diverse knapsack}---is inspired by the Chamberlin--Courant rule~\cite{cha-cou:j:cc} from the literature on multiwinner voting. Informally speaking, in this approach we aim at maximizing the number of voters who have \emph{at least one} preferred item in the selected knapsack. For the second approach---which is the main focus of the paper and which we call selecting a \emph{fair knapsack}---we use the concept of Nash welfare~\cite{nashWelfare50} from the literature on fair allocation. Nash welfare is a well-established solution concept that implements a tradeoff between having an objectively efficient resource allocation (knapsack, in our case), and having an allocation which is acceptable for a large population of agents. Indeed, extensive recent studies in the domain of fair allocation confirm particularly strong fairness guarantees of Nash welfare~\cite{CaragiannisKMP016,MoulinFair03,DarSch15}, and this solution concept has been applied e.g.\ in the context of public decision making~\cite{CFS17}, online resource allocation~\cite{FreemanZC17}~or transmission congestion control~\cite{Kelly97chargingand} (therein referred to as \emph{proportional fairness}). Thus, our work introduces a new application domain---now the goal is to select a set of shared items---for the concept of Nash welfare. In particular, as a side note, we will explain that our approach leads to a new class of multiwinner rules, which can be viewed as generalizations of the Proportional Approval Voting rule beyond the approval setting.   

Apart from introducing the new class of multiagent knapsack problems, our contributions are as follows: 
\setlist[enumerate,1]{leftmargin=0.5cm}
\setlist[enumerate,2]{leftmargin=1cm}
\begin{enumerate}[(1)]
\item We study the complexity of computing an optimal individually best (IB), diverse, and fair knapsack. This problem is in general $\np$-hard, except for the case of IB knapsack with unarily encoded utilities of the voters (as the IB knapsack problem is equivalent to the classic knapsack problem).
\item We study the parameterized complexity of computing a diverse and fair knapsack, focusing on the number of voters.\footnote{Considering this parameter is relevant, e.g., for the case when the set of voters is in fact a relatively small group of experts acting on behalf of a larger population of agents. 
}
We show that for unary-encoded utilities of the voters computing a diverse knapsack is fixed-parameter tractable when parameterized by the number of voters. 
On the contrary, computing a fair knapsack is $\wone$-hard for the same parameter.
\item We study the complexity of the considered problems for single-peaked and single-crossing preferences. 
We show that (under unary encoding of voters' utilities) a diverse knapsack can be computed in polynomial time when the preferences are single-peaked or single-crossing, while computing a fair knapsack remains $\np$-hard.     
\end{enumerate}

Our results are summarized in \Cref{tab:results}. 
Our main message is that fairness comes with a surprisingly high computational complexity. 
Indeed, our most unexpected results are that computing a fair knapsack is $\wone$-hard when parameterized by the number of voters, and $\np$-hard on single-peaked single-crossing domains, with unit-costs, and all utilities coming from the set $\{0, \ldots ,6\}$. 
This was unforeseen since by using a recent result of Peters~\cite{Pet17a}, one can show that computing a fair knapsack on single-peaked domains (which are not necessarily single-crossing, i.e., when one of the assumptions is weakened), with unit-costs, and all utilities coming from the set $\{0, 1\}$ (instead of $\{0, \ldots, 6\}$, i.e., when another assumption is strengthened) is polynomial-time solvable. 
Our result required a complex reduction from the exact set cover problem. 

Most of the our results are presented for the costs and utilities of the agents given in the unary-encoding. This makes our results more relevant for practical applications of our framework. E.g., for participatory budgeting (PB) one does not really need more than thousands of values to represent utilities/costs. Further, by assuming efficient encoding of the utilities/costs we would make the hardness results less meaningful---the hardness would simply be an artifact of the fact that we admit the values of the utilities/costs that are exponentially large in the number of voters. By assuming unary encoding we make---on the one hand---the hardness results stronger, and---on the other hand---more applicable to the real scenarios that our model represents. 

We also show all three problems to be $\np$-hard in the non-unary case (\cref{thm:dknapNPhNonunar,thm:nphardfairknaps}) and prove $\operatorname{W}$-hardness with respect to the budget (\cref{thm:lbdk,cor:wonefk}).

\newcommand{\smtab}[1]{\small#1}
\renewcommand{\arraystretch}{1.2}
\begin{table}[t]
  \centering
  \caption
  {
    Overview of our results (for the case of utilities encoded in unary): 
    Herein, SP and SC abbreviate single-peaked and single-crossing preferences, respectively, and ``\# voters'' refers to ``when parameterized by the number of voters''.
  }
  \begin{tabular}{@{}l|c|c|c|c@{}}  \toprule
				Knapsack	 & general 	& SP 	& SC 	& \#~voters \\ \cmidrule{1-5}
	      IB		& \multicolumn{3}{c|}{P~\smtab{(equivalent to \textsc{Knapsack})}}	& $\leftarrow$		\\%
	      Diverse		& $\np$-hard~\cite{pro-ros-zoh:j:proportional-representation}		& P	& P 	& FPT  	\vspace{-3pt}\\
					& %
					& \smtab{(\Cref{thm:dknapPolyUnarSP})}	&	\smtab{(\Cref{thm:dknapPolyUnarSP})}	& \smtab{(\Cref{thm:dkfpt})} \\%
	      Fair		& $\np$-hard & \multicolumn{2}{c|}{$\np$-hard} & W[1]-hard \vspace{-3pt}\\
					& \smtab{(\Cref{thm:nphardfairknaps})} &  \multicolumn{2}{c|}{\smtab{(\Cref{thm:fair_knapsack_hard_sp})}} & \smtab{(\Cref{thm:fair_knapsack_wone})}\\
   \bottomrule	
  \end{tabular}
  \label{tab:results}
\end{table}

\section{The Model} \label{sec:defi}
For a pair of natural numbers $i, j \in \naturals$, $i\leq j$, by $[i, j]$ we denote the set $\{i, i+1, \ldots, j\}$. Further, we let $[j] = [1, j]$.

Let $V = \{v_1, \ldots, v_n\}$ be the set of $n$ \emph{voters} and $A = \{a_1, \ldots, a_m\}$ be the set of $m$ \emph{items}.
The voters have preferences over the items, which are represented as a \emph{utility profile} $u =$ \mbox{$(u_{i}(a)\mid i \in [n], a \in A)$}:
for each $i \in [n]$ and $a \in A$ we use $u_{i}(a)$ to denote the utility that $v_i$ assigns to $a$; this utility quantifies the extent to which $v_i$ enjoys~$a$.
We assume that all utilities are nonnegative integers.

Each item $a \in A$ comes with a cost $c(a) \in \naturals$, and we are given a global budget $B\in\naturals$.
We call a \emph{knapsack} a subset $S$ of items whose total cost does not exceed $B$, that is $c(S) = \sum_{a \in S} c(a) \leq B$.
Our goal is to select a knapsack that would be, in some sense, most preferred by the voters. 
Below, we describe three representative rules which extend the preferences of the individual voters over individual items to their aggregated preferences over all knapsacks. Each such a rule induces
a corresponding method for selecting the best knapsack. 
Our rules are rooted in concepts from the literature on fair division and on multiwinner elections:
\begin{description}
\item[Individually best knapsack:] this is the knapsack $S$ which maximizes the total utility of the voters from the selected items $u_{\ib}(S) = \sum_{a \in S} \sum_{v_i \in V} u_i(a)$. This defines perhaps the most straightforward way to select the knapsack: we call it individually best, because the formula $u_{\ib}(S)$ treats the items separately and does not take into account fairness-related issues. 
Indeed, such a knapsack can be very unfair, as discussed in the introduction.
Indeed, such a knapsack can be very unfair, illustrated by the following:
\end{description} 
\begin{example} 
Let $B$ be an integer, and consider a set of $n = B$ voters and $m = B^2$ items, all having a unit cost (\mbox{$c(a) = 1$} for each $a \in A$). Let us rename the items so that $A = \{a_{x, y}\mid x, y \in [B]\}$ and consider the following utility profile: 
\begin{align*}
u_i(a_{x,y}) =
  \begin{cases}
    L+1       & \quad \text{if} \quad i = x = 1 \\
    L                  & \quad \text{if} \quad i = x \neq 1 \\
    0                  & \quad \text{otherwise,}
  \end{cases}
\end{align*}
for some large~$L\in\naturals$. In this case, the individually best knapsack is $S_{\ib} = \{a_{1,y}\mid y \in B\}$, that is it consists only of the items liked by a single voter $v_1$. At the same time, there exists a much more fair knapsack $S_{\fair} = \{a_{x,1}\mid x \in B\}$ that for each voter $v \in V$ contains an item liked by $v$.
\end{example}

\begin{description}
\item[Diverse knapsack:] this is the knapsack $S$ that maximizes the utility $u_{\div}(S) = \sum_{v_i \in V} \max_{a \in S} u_i(a)$. 
In words, in the definition of $u_{\div}$ we assume that each voter cares only about his or her most preferred item in the knapsack. 
This approach is inspired by the Chamberlin--Courant rule~\cite{cha-cou:j:cc} for multiwinner elections and by classic models in facility location~\cite{far-hek:b:facility-location}. We call such a knapsack diverse following the convention from the multiwinner literature~\cite{FSST-trends}. Intuitively, such a knapsack represents the diversity of the opinions among the population of voters; in particular, if the preferences of the voters are very diverse, such a knapsack tries to incorporate the preferences of as many groups of voters as possible at the cost of containing only one representative item for each ``similar'' group.    

\item[Fair knapsack:] we use Nash welfare~\cite{nashWelfare50} as a solution concept for fairness. 
Formally, we call a knapsack $S$ fair if it maximizes the product $u_{\fair}(S) = \prod_{v_i \in V} \left(1 + \sum_{a \in S}u_i(a) \right)$.\footnote{Typically, Nash welfare would be defined as $\prod_{v_i \in V} \left(\sum_{a \in S}u_i(a) \right)$. In our definition, we add one to the sum $\sum_{a \in S}u_i(a)$ in order to avoid pathological situations when the sum is equal to zero for some voters. This also allows us to represent the expression we optimize as a sum of logarithms, and thus, to expose the close relation between the fair knapsack and the Proportional Approval Voting rule. When the utilities are normalized our definition results in better properties of the outcome pertaining to fairness~\cite{FMS18}. Further, all our hardness results can be formulated for a weaker (and perhaps the least disputable) notion of fairness in the following way: it is hard to decide whether an instance of the collective knapsack problem admits a solution where the sum of the utilities of all the agents is the highest (among all valid solutions) and all the agents have the same utility.} 
Alternatively, by taking the logarithm of $u_{\fair}$ we can represent fair knapsack as the one maximizing $\sum_{v_i \in V} \log(1 + \sum_{a \in S}u_i(a))$. The following example intuitively explains the type of fairness guaranteed by using the Nash welfare.   
\end{description} 

\begin{example}\normalfont
Consider six groups of voters, $V_1, \ldots, V_6$, with
\begin{align*}
&|V_1| = 300,\,|V_2| = 200,\, |V_3| = 100,\, \text{and } \\ &  |V_4| = |V_5| = |V_6| =1\text{.}
\end{align*}
Assume we have six sufficiently large groups of items, $A_1 \ldots, A_6$. Each voter from group $V_i$ assigns utility 1 to all items from group $A_i$, and zero to all other items. Finally, assume that the costs of all items are equal to 1, and that the value of the budget is 6. An individually best knapsack would consist only of the items from $A_1$. A diverse knapsack would contain one item from each set $A_i$ for $i \in \{1,\ldots,6\}$. A fair knapsack would contain 3 items from $A_1$, 2 items from $A_2$ and 1 item from $A_3$---this solution can be interpreted as fair since the number of items selected from each group is proportional to the number of voters liking items from these groups.

If we assume that the costs of the items from groups~$A_1$,~$A_2$, and~$A_3$ are equal to 3, 2, and 1, respectively, and that our budget is equal to 6, then the fair knapsack will consist of one item from each of the sets $A_1$, $A_2$, and $A_3$---this shows that each group will obtain a share of the cost of the whole knapsack which is proportional to its size.  
\end{example}

In \cref{sec:intro} we referred to the literature supporting the use of Nash welfare in various settings. Let us complement these arguments with one additional observation.  
When the utilities of the voters come from the binary set $\{0, 1\}$ and the costs of all items are one, then our multiagent knapsack framework boils down to the standard multiwinner elections model with approval preferences. 
In this case, a very appealing rule, Proportional Approval Voting (PAV), can be expressed as finding a knapsack maximizing $\sum_{v_i \in V} H(\sum_{a \in S}u_i(a))$, where $H(i)$ is the $i$-th harmonic number. This is almost equivalent to finding a fair knapsack (maximizing the Nash welfare) since the harmonic function can be viewed as a discrete version of the logarithm. 
Thus, fair knapsack can be considered an adjustment of PAV to the model with cardinal utilities and costs. In particular (as a side note), observe that the notion of fair knapsack combined with positional scoring rules induces rules that can be viewed as adaptations of PAV to the ordinal model.    

\section{Related Work}

Our work extends the literature on the multi-objective (MO) knapsack problem~\cite{KellererKnapsack}, that is on the variant of the classic knapsack problem with multiple independent functions valuating the items. 
Typically, in the MO knapsack problem the goal is to find a (the set of) Pareto optimal solution(s) according to multiple objectives defined through given functions valuating items. Our approach is different since we consider specific forms of aggregating the objectives; in particular for each of the concepts we study---for the individually best, diverse, and fair knapsack---there always exists a Pareto optimal solution; further each solution to the individually best and fair knapsack is Pareto optimal. For an overview of the literature on the MO knapsack problem (with the focus on the analysis of heuristic algorithms) we refer to the survey by Lust and Teghem~\cite{LustT12}.

Lu and Boutilier~\cite{bou-lu:c:chamberlin-courant} studied a variant of the Chamberlin--Courant rule that includes knapsack constraints and so being very similar to our diverse knapsack problem. The difference is that
\begin{inparaenum}[(i)]
\item they consider utilities which are extracted from the voters' preference rankings, thus these utilities have a specific structure, and
\item in their model the items are not shared; instead, the selected items can be copied and distributed among the voters. Lu and Boutilier consider a model with additional costs related to copying a selected item and sending it to a voter.
\end{inparaenum}
Consequently, their general model is more complex than our diverse knapsack; they also considered a more specific variant of this model, equivalent to winner determination under the Chamberlin--Courant rule.

A variant of the diverse knapsack problem with the utilities satisfying a form of the triangle inequality is known as the knapsack median problem; see the work of Byrka~et~al.~\cite{ByrPRSST15} for a discussion on the approximability of the problem.

As we discussed in the introduction, the multiagent variant of the knapsack problem has been often considered in the context of participatory budgeting, yet to the best of our knowledge this literature focuses on the simplest aggregation rule corresponding to our individually best knapsack approach~\cite{participatoryBudgeting,knapsackVoting,BNPS17}. Another avenue has been explored by Fain~et~al.~\cite{FainGM16}, who studied rules that determine the level of funding provided to different projects (items, in our nomenclature) rather than rules selecting subsets of projects with predefined funding requirements.

\section{Computing Multiagent Knapsacks}

In this section we investigate the computational complexity of finding individually best, diverse, and fair knapsack. Formally, we define the computational problem for computing a fair knapsack~as: 

\Optproblemdef{Fair Knapsack}
{An instance~$(V,A,u,c)$ and a budget~$B$.}
{Compute a knapsack~$S\subseteq A$ such that $c(S)\leq B$ and $u_{\fair}(S)$ is maximum.}

\noindent
We define the computational problems \textsc{Diverse Knapsack} and \textsc{Individually Best Knapsack} analogously---the difference is only in the expression to maximize, which for the two problems is $u_{\div}$ and $u_{\ib}$, respectively. We will use the same names when referring to the decision variants of these problems; in such cases we will assume that one additional integer $x$ is given in the input, and that the decision question is whether there exists $S$ with value $u_{\ib}(S)$ (respectively, $u_{\div}(S)$ or $u_{\fair}(S)$) at least~$x$, and~$c(S)\leq B$.  

We observe that the functions $u_{\ib}$, $u_{\div}$, and $u_{\fair}$ (when represented as a sum of logarithms---the use of the logarithm in the objective is only relevant for the approximation ratio) are submodular. 
Thus, we can use an algorithm of~\cite{Sviridenko04} with the following guarantees.

\begin{theorem}
There is a polynomial-time $(1 - \nicefrac{1}{e})$-approximation algorithm for \textsc{Individually Best Knapsack}, \textsc{Diverse Knapsack}, and \textsc{Fair Knapsack} with the objective function $\log(u_{\fair})$.
\end{theorem}

In the remainder of the paper we will focus on computing exact solutions for the three problems. In particular, we study the complexity under the following two restricted domains:
\begin{description}
\item[Single-peaked preferences.] Let $\top_i$ denote $v_i$'s most preferred item, and let $\triangleleft$ be an order of the items. We say that a utility profile $u$ is \emph{single-peaked} with respect to $\triangleleft$ if for each $a, b \in A$ and each $v_i \in V$ such that $a \triangleleft b \triangleleft \top_i$ or $\top_i \triangleleft b \triangleleft a$ we have that $u_i(b) \geq u_i(a)$.   
\item[Single-crossing preferences.] Let $\triangleleft$ be an order of the voters. We say that a utility profile $u$ is \emph{single-crossing} with respect to $\triangleleft$ if for each two items $a, b \in A$ the set $\{v_i \in V\mid u_i(b) \geq u_i(a)\}$ forms a consecutive block according to $\triangleleft$.
\end{description}

We say that a profile $u$ is single-peaked (resp., single-crossing) if there exists an order~$\triangleleft$ of the items (resp., of the voters) such that~$u$ is single-peaked (resp., single-crossing) with respect to $\triangleleft$. Note that an order witnessing single-peakedness or single-crossingness can be computed in polynomial time (see, e.g.,~\cite{ELP-trends,fit15}).

We will also study the parameterized complexity of the three problems. 
For a given parameter $p$, we say that a problem is fixed-parameter tractable (FPT) when parameterized by~$p$ if there is an algorithm (FPT algorithm) that solves each instance $I$ of the problem in~$O\big(f(p) \cdot \poly(|I|)\big)$~time, where $f$ is some computable function.
In parameterized algorithmics, FPT algorithms are considered efficient. 
There is a whole hierarchy of complexity classes, but informally speaking, a problem that is~$\wone$- or~$\wtwo$-hard is assumed not to be FPT and, hence, hard (or fixed-parameter intractable) from the parameterized point of view (see~\cite{dow-fel:b:parameterized} for more details).

\subsection{Individually Best Knapsack}

We start with the simplest case of individually best knapsack.

\begin{theorem}
  \label{thm:ieknapPoly}
 \textsc{Individually Best Knapsack} is solvable in polynomial time when the utilities of voters are unary-encoded.
\end{theorem}

\begin{proof}
 Consider an instance~$(V=\{v_1,\ldots,v_n\},A=\{a_1,\ldots,a_m\},u,c,B)$, and let~$\hat{u}:=\sum_{v_i\in V} \sum_{a\in A} u_i(a)$.
 We apply dynamic programming with table~$T$, where~$T[i,x]$ denotes the minimal cost of $S \subseteq \{a_1,\ldots,a_i\}$ with value $u_{\ib}(S)$ \emph{at least} equal to~$x$.
 We initialize~$T[i,0]=0$ for $i \in [m]$ and~$T[0,x]=\infty$ for each~$x \in [\hat{u}]$.
 We define the helper function~$f(i,x):=x-\sum_{v_j\in V} u_j(a_i)$, where~$i\in[m]$ and~$x \in [\hat{u}]$.
 For $i\in[m]$ and~$x \in [\hat{u}]$, we have
 \[ T[i,x] = \min\left(\begin{array}{l}
                        T[i-1,x], \\
                        c(a_i) + T\left[i-1,\max(0,f(i,x))\right]
                       \end{array}
	      \right).
  \]
  By precomputing~$\sum_{v_j\in V} u_j(a)$ for each~$a\in A$, we get a running time of~$O(nm+ m\hat{u})$.
\end{proof}

Note that if the utilities are not encoded in unary, then \textsc{Fair Knapsack} is $\np$-hard even for one voter (see~\cref{thm:nphardfairknaps}).

\subsection{Diverse Knapsack}

We now turn our attention to the problem of computing a diverse knapsack. 
Through a straightforward reduction from the standard knapsack problem, we get that \dknap{} is computationally hard even for profiles which are both single-peaked and single-crossing, unless the utilities are provided in unary encoding.

\begin{theorem}
  \label{thm:dknapNPhNonunar}
 \textsc{Diverse Knapsack} is $\np$-hard even for single-peaked and single-crossing utility profiles.
\end{theorem}
\begin{proof}
  We present a many-one reduction from~\textsc{Knapsack}.
  Let~$(X=\{x_1,\ldots,x_n\},x,y)$ be an instance of \textsc{Knapsack} where each~$x_i$ comes with value~$\nu(x_i) \geq 1$ and weight~$\omega(x_i)$;
  the question is whether there exists $S \subseteq X$ with $\sum_{x_i \in S} \nu(x_i) \geq x$ and $\sum_{x_i \in S} \omega(x_i) \leq y$.
  We set our set of items~$A=\{a_1,\ldots,a_n\}$ with~$c(a_i):=\omega(x_i)$ for each $i\in [n]$.
  We add~$n$ voters~$v_1,\ldots,v_n$ with
  \begin{align*}
  u_i(a_j):=
  \begin{cases}
       3n^2 \cdot \nu(a_j),& i=j,\\
       j,& i > j, \\
       2n-j + 1,&  i < j.
  \end{cases}
  \end{align*}
  It is immediate that for each $S$ we have that $\sum_{a_i \in S} c(a_i) = \sum_{x_i \in S} \omega(x_i)$. 
  Further, $\sum_{v_i \in V} \max_{a_j \in S}u_i(a_j) \geq 3n^2 x$ if and only if $\sum_{x_j \in S} \nu(x_j) \geq x$, which proves the correctness. 
  It is immediate to check that the utility profile is single-peaked and single-crossing. 
\end{proof}

Note that \dknap{} is~$\np$-hard for utilities encoded in unary as it generalizes the Chamberlin--Courant rule, which is computationally hard~\cite{pro-ros-zoh:j:proportional-representation}. 
For single-peaked or single-crossing profiles the Chamberlin-Courant rule is computable in polynomial time~\cite{bet-sli-uhl:j:mon-cc,sko-yu-fal-elk:j:sc-cc}. 
These known algorithms can be extended 
to the case of the diverse knapsack. 

\begin{theorem}
  \label{thm:dknapPolyUnarSP}
 \textsc{Diverse Knapsack} is solvable in polynomial time when the utility profile is%
 \begin{enumerate}[(i)]
  \item single-peaked and encoded in unary;
  \item single-crossing and encoded in unary.
 \end{enumerate}
\end{theorem}

\begin{proof}[Proof of (i)]
  Consider an input instance~$(V=\{v_1,\ldots,v_n\},A=\{a_1,\ldots,a_m\}, u, c, B)$, where~$A$ is enumerated such that the order is single-peaked (note that such an ordering can be computed in polynomial time~\cite{ELP-trends}).
  Let~$\hat{u}:=\sum_{v_i\in V} \sum_{a\in A} u_i(a)$.
  We apply dynamic programming with table~$T$, where~$T[i,x]$ denotes the minimal cost of a subset $S \subseteq \{a_1,\ldots,a_i\}$ containing~$a_i$ ($a_i \in S$) with value \emph{at least} equal to~$x$ ($u_{\div}(S) \geq x$).
  We define the helper function~
  $$f(i,x):=\begin{cases}
	    c(a_i),& \text{if }\sum_{v_j\in V} u_j(a_i)\geq x,\\
	    \infty,& \text{otherwise}.
	  \end{cases}$$
  We initialize~$T[1,x]=f(1,x)$ for all~$x\in[0,\hat{u}]$.
  Then we set
  \[ T[i,x] = \min\left(\begin{array}{l}
                        f[i,x], \\
                        c(a_i) + \min_{1\leq j< i} T\left[j,x-d(i,j)\right]
                       \end{array}
	      \right),
  \]
  where $d(i,j):=\sum_{v_\ell\in V} \max(0,u_\ell(a_i)-u_\ell(a_j))$.
  Let~$M:=\{(i,x)\mid T[i,x]\leq B\}$.
  Then we can derive the value of the best diverse knapsack from~$\max\{x\mid \exists i\colon (i,x)\in M\}$.
  Clearly, when the utilities are unarily encoded, we can compute all entries of~$T$ and set~$M$ in polynomial time.
  
  We inductively argue over~$1\leq i\leq n$.
  Clearly, the best diverse knapsack over item set~$\{a_1\}$ has cost~$c(a_1)$.
  Consider~$T[i,x]$ with~$i>1$.
  Let~$A_i\subseteq\{a_1,\ldots,a_i\}$ be a set of items with value at least~$x$ of minimal cost containing~$a_i$.
  Then, either~$A_i=\{a_i\}$, or $A_i=A_j\cup\{a_i\}$, where~$j\in[i-1]$ with~$a_j\in A_i$ and there is no~$j'\in[i-1]$ such that $a_{j'} \in A_i$ and~$j<j'<i$.
  In the first case,~$c(A_i)=c(a_i)=f[i,x]$.
  Consider the second case.
  Clearly, $c(A_j)=c(A_i)-c(a_i)$.
  Let~$V_1 = \big\{v_{\ell} \in V\mid u_\ell(a_j)\leq u_{\ell}(a_i)\big\}$ and let~$V_2:=V\setminus V_1$.
  From the single-peakedness, we have that~$u_\ell(a_i)-\max_{1\leq j'\leq j} u_\ell(a_{j'}) = u_\ell(a_i) - u_\ell(a_j)$, if~$v_\ell\in V_1$, and clearly~$u_\ell(a_i)-\max_{1\leq j'\leq j}u_\ell(a_{j'}) <0$, if~$v_\ell\in V_2$.
  Hence, the value of~$A_i$ is greater than the value of $A_j$ by~$\sum_{v_\ell\in V} \max(0,u_\ell(a_i)-u_\ell(a_j))=d(i,j)$.
\end{proof}

\cref{thm:dknapPolyUnarSP}(i) is proven via straight-forward dynamic programming and omitted due to space constraints.
We prove our result for single-crossingness.
Let us define a set of useful tools. We will also use these tools later on, when analyzing the parameterized complexity of the problem.   

Given a tuple of voters~$\vec{V}=(v_1,\ldots,v_n)$ and a subset~$S\subseteq A$ of items, we define an~\emph{assignment}~$\pi_{S,\vec{V}}$ as a surjection~$[n]\to S$.
An assignment is called \emph{connected}, if for every~$s\in S$ it holds that $\pi_{S,\vec{V}}^{-1}(s):=\{i\in [n]\mid s=\pi_{S,\vec{V}}(i)\}=[x,y]$ for some~$x,y\in [n]$, $y \geq x$.
For our first tool we introduce the following auxiliary problem.
\medskip
\Optproblemdef{Ordered Diverse Knapsack}
{An instance~$(\vec{V},A,u,c)$ where $\vec{V}=(v_1,\ldots,v_n)$ is ordered and a budget~$B$.}
{Compute a knapsack~$S\subseteq A$ such that $c(S)\leq B$, and $u_{\rm Ord}(S)=\max_{\text{connected~}\pi_{S,\vec{V}}} \sum_{i=1}^n u_i(\pi_{S,\vec{V}}(i))$ is maximum.}
\smallskip
If~$S=\{s_1,\ldots,s_\ell\}\subseteq A$ is a cost-minimal solution to~\textsc{Diverse Knapsack} on~$(V,A,u,c)$, then let~$V_i:=\{v_j\in V\mid s_i \in \arg\max_{a\in S} u_j(a)\}$.
Consider an ordering~$\vec{V}=(V_1,\ldots,V_\ell)$, where for each~$i\in[\ell]$, the voters in~$V_i$ are arbitrarily ordered.
Then it is not difficult to see that the assignment
$$\pi_{S,\vec{V}}(i)=\arg\max_{a\in S} u_i(a)$$ is connected.
Hence, we obtain the following connection between \textsc{Diverse Knapsack} and \textsc{Ordered Diverse Knapsack}.

\begin{observation}
 \label{obs:odkeqdk}
 There is an ordering~$\vec{V}$ on the voters~$V$ such that there is an~$S\subseteq A$ that forms a cost-minimal solution for \textsc{Ordered Diverse Knapsack} and for~\textsc{Diverse Knapsack}.
\end{observation}

Next, we give a dynamic program for computing knapsacks that qualitatively lie ``between'' optimal solutions for \textsc{Ordered Diverse Knapsack} and~\textsc{Diverse Knapsack} (what we mean by ``lying in between'' is specified later on).

Let us fix an input~$(V,A,u,c, B)$ and an ordering~$\vec{V}=(v_1,\ldots,v_n)$ of the voters. 
We set~$\hat{u}:=\sum_{i=1}^n \sum_{a\in A} u_i(a)$.
We give a dynamic program with table~$T$, where $T[i,x]$ denotes ``some'' cost of a knapsack with a value assigned by voters from $(v_1,\ldots,v_i)$ at least equal to~$x$.
We set $$T[1,x]=\min\{c(a)\mid a\in A, u_1(a)\geq x\},$$ if there is an $a\in A$ such that $u_1(a)\geq x$, and $T[1,x]=\infty$ otherwise.
We define a helper function
\[ f(i,a,x)=\begin{cases} c(a),& \text{if~~}\sum_{j=1}^{i}u_j(a)\geq x,\\\infty,&\text{otherwise.}\end{cases}\]
We set
\begin{align}
T[i,x]&=\min_{a\in A}
	  \left(\begin{array}{l}
	    f(i,a,x),\\
	    c(a)+\min\limits_{\mathclap{j\in[i-1]}}\ T[j,\max(0,x-\sum\limits_{\mathclap{\ell=j+1}}^i u_\ell(a))]
	    \end{array}
	  \right).
  \label{eq:tabT}
\end{align}
  
\begin{observation}
  \label{obs:Tpt}
  When the utilities are unarily encoded, we can compute all entries of~$T$ in polynomial time.
\end{observation}

\begin{lemma}
  \label{lemma:Tdk}
 Let~$S$ be a cost-minimal solution to~\textsc{Diverse Knapsack} on~$(V,A,u,c)$ and let~$x=u_{\rm Div} (S)$.
 Then~$T[n,x]\geq c(S)$.
\end{lemma}

\begin{proof}
  Suppose that this is not the case, that is, $T[n,x]<c(S)$.
  Then we construct a knapsack~$S'$ from~$T[n,x]$ as follows.
  Let~$a\in A$ be an item that minimizes~\eqref{eq:tabT} for~$T[n,x]$, then make~$a\in S'$.
  If~$T[n,x]=f(n,a,x)$, then~$T[n,x]=c(a)< c(S)$, contradicting the fact that~$S$ is cost-minimal.
  Otherwise,
  \begin{align*}
  T[n,x]=c(a)+T[j,x':=\max(0,x-\sum\limits_{\ell=j+1}^n u_\ell(a))]\text{,}
  \end{align*}
  for some~$j\in[n-1]$. Then we proceed towards a contradiction as before: 
  Let~$a'\in A$ be an item that minimizes~\eqref{eq:tabT} for~$T[j,x']$, then make~$a'\in S'$, and continue the same reasoning.
\end{proof}

\begin{lemma}
  \label{lemma:Todk}
 Let~$S$ be a cost-minimal solution to~\textsc{Ordered Diverse Knapsack} on~$(\vec{V},A,u,c)$ where $\vec{V}$ is ordered and let~$x=u_{\rm Ord} (S)$.
 Then~$T[n,x]\leq c(S)$.
\end{lemma}

\begin{proof}
  Assume $S=\{s_1,\ldots,s_\ell\}$ being enumerated.
  Let $\pi_{S,\vec{V}}$ be a connected assignment such that~$\sum_{v_i \in V} u_i(\pi_{S,\vec{V}}(i))=x$.
  Let~$0=i_0<i_1<\ldots<i_\ell=n$ be such that $[i_j+1,i_{j+1}]=\pi_{S,\vec{V}}^{-1}(s_{j+1})$ for every~$0\leq j<\ell$.
  Moreover, let~$x_j:=\sum_{i=1}^{i_j} u_i(\pi_{S,\vec{V}}(i))$, for~$j \in [\ell]$.
  By our definition of~$T$, we have that~$T[i_1,x_1]\leq c(s_1)$.
  Moreover, we have~$T[i_2,x_2]\leq T[i_1,x_1]+c(s_2)$.
  It follows inductively that~$T[i_\ell,x_\ell]\leq c(S)$.
\end{proof}

We have all ingredients at hand to prove our main results.

\begin{proof}[Proof of \cref{thm:dknapPolyUnarSP}(ii)]
If~$\vec{V}$ is an order witnessing single-crossingness, then there is an~$S\subseteq A$ that forms a cost-minimal solution for \textsc{Ordered Diverse Knapsack} and for~\textsc{Diverse Knapsack}. 
\cref{lemma:Tdk,lemma:Todk} guarantee that our algorithm will find it.
\end{proof}

Further, we can use our tools to obtain an FPT algorithm (for the number of voters) for unrestricted domains.

\begin{theorem}
  \label{thm:dkfpt}
 \textsc{Diverse Knapsack} is FPT when parameterized by the number of voters and the utilities are unarily encoded.
\end{theorem}
\begin{proof}
 By~\cref{obs:odkeqdk}, we know that there is an order~$\vec{V}$ on the voters such that there is an~$S\subseteq A$ that forms a cost-minimal solution for \textsc{Ordered Diverse Knapsack} and for~\textsc{Diverse Knapsack}.
 Together with~\cref{lemma:Tdk,lemma:Todk} we obtain that for~$\vec{V}$ our dynamic program will find such~$S$.
 Hence, for each ordering of~$V$, we compute~$T[n,x]$.
 Then, we take the minimum over all observed values.
 Note that~$x$ is the largest value such that~$T[n,x]\leq B$ for some ordering of the voters.
 Altogether, this yields a running time of~$O(n!\poly(\hat{u}+n+m))\subseteq O(2^{n\log n}\poly(\hat{u}+n+m))$.
\end{proof}

Finally, we complement~\cref{thm:dkfpt} by proving a lower bound on the running time (via reducing from the $\wtwo$-complete \textsc{Dominating Set} problem), assuming the Exponential-Time Hypothesis (ETH).

\begin{proposition}
  \label{thm:lbdk}
 \textsc{Diverse Knapsack} with binary utilities and unary costs is~$\wtwo$-hard when parameterized by~the budget~$B$ and, unless the ETH breaks, there is no~$2^{o(n+m)}\cdot \poly(n+m)$ algorithm.
\end{proposition}

\begin{proof}
 We give a many-one reduction from \textsc{Dominating Set}. 
 An instance of dominating set consists of a graph $G=(W,E)$ and an integer $k$; 
 the question is whether there exists a subset~$S$ of at most~$k$ vertices such that for each vertex~$w\in W$ there is an~$s\in S$ such that~$w\in N_G[s]$, where $N_G[s]=\{v\in W\mid \{v,s\}\in E\}\cup\{s\}$ denotes the closed neighborhood of~$s$ in~$G$.
 For each vertex $w \in W$, we introduce a voter~$v_w$ to~$V$ and an item~$a_w$ to~$A$ of cost one. We set~$u_{v_w}(a_{w'})=1$, if~$w'\in N_G[w]$, and $u_{v_w}(a_{w'})=0$, otherwise.
 Furthermore, we set the budget~$B=k$.
 It is not difficult to see that there is a diverse knapsack~$S$ with~$c(S)\leq B$ and~$u_{\rm Div}(S)\geq n$ if and only if~$(G,k)$ is a yes-instance.
 As~$B=k$ and~$|V|=|A|=n$, the lower bounds follow.
\end{proof}

\subsection{Fair Knapsack}

Let us now turn to the problem of computing a fair knapsack.
We first prove that the problem is $\np$-hard, even for restricted cases, and then we study its parameterized complexity.

\begin{theorem}
 \label{thm:nphardfairknaps}
\fknap{} is $\np$-hard, even
 \begin{enumerate}
  \item for one voter;
  \item for two voters and when all costs are equal to one;
  \item if all utilities are in~$\{0,1\}$ and all costs are equal to one.
 \end{enumerate}
\end{theorem}

\begin{proof}

(1): We provide a many-one reduction from the \textsc{Partition} problem~\cite{gar-joh:b:int}: Given a set~$S=\{s_1,\ldots,s_n\}$ of~$n$ positive integers, the question is to decide whether there exists a subset~$S'\subseteq S$ such that $\sum_{s\in S'} s=\frac{1}{2}\sum_{s\in S} s$. 
Given an instance $(S)$ of \textsc{Partition} where all integers are divisible by two, we construct an instance of \fknap{} as follows.
Let $T:=\sum_{s\in S} s$.
For each $s_i\in S$, we introduce an item $a_i$ with cost~$s_i$.
Further, we introduce one voter~$v_1$ with utility $u_1(a_i)=s_i$ for each $i\in[n]$.
We set the budget~$B=T/2$, and we ask if there exists a knapsack with a Nash~welfare~$W$ of at least~$T/2+1$.

Let $(S)$ be a yes-instance and let $S'\subseteq S$ be a solution.
Then, the subset of items $A':=\{a_i\in A\mid s_i\in S'\}$ forms a fair knapsack, as $\sum_{a\in A'} c(a)=\sum_{s\in S'} s= T/2\leq B$, and the Nash welfare is at least $1+\sum_{a\in A'} (u_1(a))=1+\sum_{s\in S'} s= T/2+1$

Conversely, let the constructed instance of \fknap{} be a yes-instance, and let~$A'\subseteq A$ be a fair knapsack.
Denote by~$S'$ the subset of integers in~$S$ corresponding to the items in~$A'$.
Then it holds that $\sum_{s\in S'} s=\sum_{a\in A'} c(a)\leq T/2$.
Moreover, $1+\sum_{s\in S'} s=1+\sum_{a\in A'} (u_1(a))\geq T/2+1$.
Together, both inequalities yield~$\sum_{s\in S'} s=T/2$, and hence~$S'$ forms a solution to~$(S)$.

\smallskip
\noindent
(2): We provide a many-one reduction from the \textsc{Exact Partition} problem: Given a set~$S=\{s_1,\ldots,s_n\}$ of~$n$ positive integers and an integer $k\geq 1$, decide whether there is a subset~$S'\subseteq S$ with $|S'|=k$ such that $\sum_{s\in S'} s=\frac{1}{2}\sum_{s\in S} s$.
Given an instance $(S,k)$ of \textsc{Exact Partition} where all integers are divisible by two and by~$k$, we construct an instance of~\fknap{} as follows.
Similarly as before we set $T := \sum_{s\in S} s$. 
For each $s_i\in S$, we introduce an item $a_i$ with cost~$1$.
Further, we introduce two voters, $v_1$ and~$v_2$, with utility functions~$u_1(a_i)=T+s_i$ and~$u_2(a_i)=T+\frac{T}{k}-s_i$, $1\leq i\leq n$, respectively.
We set the budget~$B=k$ and ask for a knapsack with a Nash~welfare~$W$ at least equal to~$(1 + kT + T/2)^2$.

Let $(S,k)$ be a yes-instance and let $S'\subseteq S$ be a solution.
We claim that the subset of items $A':=\{a_i\in A\mid s_i\in S'\}$ forms an appropriate fair knapsack.
It holds that~$\sum_{a\in A'} c(a)=|S'|=k\leq B$.
The Nash welfare is at least 
\begin{align*}
\left(1+\sum_{a\in A'} u_1(a)\right)\cdot \left(1+\sum_{a\in A'} u_2(a)\right)
&=\left(1+kT+\sum_{s\in S'} s\right) \cdot \left(1+kT+\sum_{s\in S'} \left(\frac{T}{k}-s\right)\right) \\
&=\left(1+kT+ T/2\right)\left(1 +kT + T-\sum_{s\in S'} s\right)\\
&=\left(1 + kT + T/2\right)^2. 
\end{align*}

Conversely, let the constructed instance of \fknap{} be a yes-instance, and let~$A'\subseteq A$ be a corresponding fair knapsack.
Let~$S'$ denote the subset of integers in~$S$ corresponding to the items in~$A'$.
Then it holds true that $|S'|=\sum_{a\in A'} c(a)\leq B=k$.
Moreover, for each item~$a\in A$ it holds true that~$u_1(a)+u_2(a)=2T + T/k$, and hence $\sum_{a\in A'}(u_1(a)+u_2(a))=|A'|(2T + T/k)$.
The product~$(1+\sum_{a\in A'} u_1(a))\cdot (1+\sum_{a\in A'} u_2(a))$ is maximal if~$\sum_{a\in A'}u_1(a)=\sum_{a\in A'}u_2(a)$, leading to~$\sum_{a\in A'}u_1(a)=|A'|\left(T + \frac{T}{2k}\right)$.
Together, it follows that~$|A'|=k$, and hence~$S'$ forms a solution for~$(S,k)$.
\smallskip

\noindent
(3): We provide a many-one reduction from the \textsc{Exact Regular Set Packing (ERSP)} problem (there is a straight-forward parameterized reduction from \textsc{Exact Regular Independent Set}~\cite{aus-atr-pro:j:structure}): 
Given a set~$X$, set~$\calF=\{F_1,\ldots,F_m\}$ of subsets of~$X$ with~$|F_i|=d$ for all $i\in[m]$, and an integer $k\geq 1$, decide whether there exists a subset~$\calF'\subseteq \calF$ with $|\calF'|= k$ such that for each distinct~$F,F'\in \calF$ it holds true that~$F\cap F'=\emptyset$.
Let $(X=\{x_1,\ldots,x_n\},\calF=\{F_1,\ldots,F_m\},k)$ be an instance of \textsc{ERSP} where $|F_i|=d$ for all $i\in[m]$.
We construct an instance of~\fknap{} as follows.
Let $A:=\{a_i\mid F_i\in \calF\}$ be the set of items each with cost equal to one.
Further, we introduce~$n$~voters with $u_i(a_j)=1$ if $x_i\in F_j$, and $u_i(a_j)=0$ otherwise, for all~$i\in[n]$, $j\in[m]$.
We set~$B=k$ and the desired Nash welfare to $W=2^{dk}$.
This finishes the construction.

Assume that $(X,\calF,k)$ admits a solution~$\calF'$.
We claim that~$A':=\{a_i\in A\mid F_i\in \calF'\}$ is a fair knapsack with the desired value of the Nash welfare (note that~$|A'|=k$).
By the construction, each item~$a\in A$ contributes one to exactly~$d$ voters.
Moreover, each distinct $a,a'\in A'$ contribute to disjoint sets of voters.
Hence,
\begin{align*}
\prod_{1\leq i\leq n}\left(1+\sum_{a\in A'}u_i(a)\right)=2^{dk} \text{.}
\end{align*}

Conversely, let~$A'\subseteq A$ be a fair knapsack, and
let $\calF'=\{F_i\in \calF\mid a_i\in A'\}$.
We claim that $\calF'$ forms a solution to~$(X,\calF,k)$.
First, observe that~$\sum_{1\leq i\leq n}\sum_{a\in A'}u_i(a)=|A'|d$.
Let $M:=\{i\in[n]\mid x_i\in \bigcup \calF'\}$ be the set of elements covered by $\calF'$.
Note that $1\leq|M|\leq |A'|\cdot d$.
Then
\begin{align*}
\prod_{i\in M}\left(1+\sum_{a\in A'}u_i(a)\right)\leq \left(1+\frac{|A'|d}{|M|}\right)^{|M|} \leq 2^{|A'|d}\leq 2^{kd}.
\end{align*}
For the second inequality, observe that the function $(1+y/x)^x$ is increasing on the interval $(0,y]$ for every~$y>0$.
Hence, we have that $|A'|=k$ and $|M|=k\cdot d$.
Thus, $\calF'$ is a set of exactly~$k$ pairwise disjoint sets.
\end{proof}

The proof of \mbox{\cref{thm:nphardfairknaps} (3)} uses the reduction from \textsc{Exact Regular Set Packing (ERSP)}. Since \textsc{ERSP} is $\wone$-hard with respect to the size of the solution~\cite{aus-atr-pro:j:structure}, we get the following corollary.

\begin{corollary}
  \label{cor:wonefk}
  \fknap{} is $\wone$-hard when parameterized by the budget, even if all utilities are in~$\{0,1\}$ and all costs are equal to one.
\end{corollary}

Using a different construction, we can show that for the combination of the two parameters---the number of voters and the budget---we still get fixed-parameter intractability.

\begin{theorem}
  \label{thm:fair_knapsack_wone}
\fknap{} is $\wone$-hard when parameterized by the number of voters and the budget, even if the utilities and the budget are represented in unary encoding and the costs of all items are equal to one.
\end{theorem}
\begin{proof}
We provide a parameterized reduction from the $k$-\textsc{Multicolored Clique} problem, which is known to be $\wone$-hard with respect to the number of colors. 
Let $I$ be an instance of $k$-\textsc{Multicolored Clique}. 
In $I$ we are given a graph $G$ with the set~$V(G)$ of vertices and the set~$E(G)$ of edges, a natural number $k \in \naturals$, and a coloring function $f\colon V(G) \to [k]$ that assigns one of~$k$~colors to each vertex. We ask if $G$ contains $k$ pairwise connected vertices, each having a different color. 
Without loss of generality we assume that~$k \geq 2$.

From~$I$ we construct an instance $I_{\mathrm{F}}$ of \fknap{} as follows (we refer to~\cref{fig:fair_knapsack_wone} for an illustration). 
\begin{figure}[t]
  \centering
    \begin{tikzpicture}
    \usetikzlibrary{decorations.pathreplacing}
    \tikzstyle{xline}=[thin, dotted, gray];
    \tikzstyle{xbrace}=[decorate,decoration={brace,amplitude=3pt}];
    \def\xr{0.65}
    \def\yr{-0.65}
    \def\xmn{0.4}
    \def\xmx{15.33}
    \def\ymn{-0.25}
    \def\ymx{13.5}

    \newcommand{\txtdec}{\scriptsize}
    \newcommand{\txtdecx}{\tiny}

    \newcommand{\cone}[6]{
    \node at (1*\xr,#1*\yr)[]{#2};
    \node at (2*\xr,#1*\yr)[scale=0.7]{#3};
    \node at (3*\xr,#1*\yr)[]{#4};
    \node at (4*\xr,#1*\yr)[scale=0.7]{#5};
    \node at (5*\xr,#1*\yr)[]{#6};
    }
    \newcommand{\conex}[7]{
    \node at (1*\xr,#1*\yr)[scale=#7]{#2};
    \node at (2*\xr,#1*\yr)[scale=#7]{#3};
    \node at (3*\xr,#1*\yr)[scale=#7]{#4};
    \node at (4*\xr,#1*\yr)[scale=#7]{#5};
    \node at (5*\xr,#1*\yr)[scale=#7]{#6};
    }
    \newcommand{\ctwo}[6]{
    \node at (6*\xr,#1*\yr)[]{#2};
    \node at (7*\xr,#1*\yr)[scale=0.7]{#3};
    \node at (8*\xr,#1*\yr)[]{#4};
    \node at (9*\xr,#1*\yr)[scale=0.7]{#5};
    \node at (10*\xr,#1*\yr)[]{#6};
    }
    \newcommand{\ctwox}[7]{
    \node at (6*\xr,#1*\yr)[scale=#7]{#2};
    \node at (7*\xr,#1*\yr)[scale=#7]{#3};
    \node at (8*\xr,#1*\yr)[scale=#7]{#4};
    \node at (9*\xr,#1*\yr)[scale=#7]{#5};
    \node at (10*\xr,#1*\yr)[scale=#7]{#6};
    }
    \newcommand{\rec}[5]{
    \def\rxy{0.4}
    \draw[rounded corners, gray, very thin, fill=white] (#1*\xr-#3*\xr-\rxy*\xr,#2*\yr-#4*\yr-\rxy*\yr) rectangle (#1*\xr+#3*\xr+\rxy*\xr,#2*\yr+#4*\yr+\rxy*\yr);
    \node at (#1*\xr,#2*\yr)[]{#5};
    }

    \draw[rounded corners, gray, thin] (\xmn*\xr,0.5*\yr) rectangle (\xmx*\xr+0.25*\xr,13.5*\yr);

    \def\hy{-0.0}
    \cone{\hy}{$n^1_1$}{$\cdots$}{$n^1_i$}{$\cdots$}{$n^1_\ell$};
    \draw[xline] (5.5*\xr,\ymn*\yr) -- (5.5*\xr,\ymx*\yr);
    \ctwo{\hy}{$n^2_1$}{$\cdots$}{$n^2_j$}{$\cdots$}{$n^2_\ell$};
    \draw[xline] (10.5*\xr,\ymn*\yr) -- (10.5*\xr,\ymx*\yr);
    \node at (11*\xr,\hy*\yr)[scale=0.7]{$\cdots$};
    \draw[thin, dashed, gray] (11.5*\xr,\ymn*\yr) -- (11.5*\xr,\ymx*\yr);
    \node at (12*\xr,\hy*\yr)[scale=0.7]{$\cdots$};
    \node at (13.5*\xr,\hy*\yr)[]{$\{n^1_i,n^2_j\}$};
    \node at (15*\xr,\hy*\yr)[scale=0.7]{$\cdots$};

    \def\hx{-0.5}
    \node at (\hx*\xr,1*\yr)[]{$v_1$};
    \node at (\hx*\xr,2*\yr)[]{$v_2$};
    \node at (\hx*\xr,3*\yr)[]{\txtdec$\vdots$};
    \draw[xline] (-1*\xr,3.5*\yr) -- (\xmx*\xr,3.5*\yr);
    \node at (\hx*\xr,4*\yr)[]{$v_{\{1,2\}}^1$};
    \node at (\hx*\xr,5*\yr)[]{\txtdec$\vdots$};
    \node at (\hx*\xr,6*\yr)[]{$v_{\{1,2\}}^{k-2}$};
    \node at (\hx*\xr,7*\yr)[]{\txtdec$\vdots$};
    \node at (\hx*\xr,8*\yr)[]{\txtdec$\vdots$};
    \draw[xline] (-1*\xr,8.5*\yr) -- (\xmx*\xr,8.5*\yr);
    \node at (\hx*\xr,9*\yr)[]{$v_{(1,2)}^a$};
    \node at (\hx*\xr,10*\yr)[]{$v_{(1,2)}^b$};
    \node at (\hx*\xr,11*\yr)[]{$v_{(2,1)}^a$};
    \node at (\hx*\xr,12*\yr)[]{$v_{(2,1)}^b$};
    \node at (\hx*\xr,13*\yr)[]{\txtdec$\vdots$};


    \cone{1}{$T$}{$\cdots$}{$T$}{$\cdots$}{$T$};
    \ctwo{2}{$T$}{$\cdots$}{$T$}{$\cdots$}{$T$};
    \node at (11*\xr,3*\yr)[scale=0.7]{$\cdots$};

    \node at (13.5*\xr,4*\yr)[]{$T$};
    \node at (13.5*\xr,5*\yr)[scale=0.7]{$\vdots$};
    \node at (13.5*\xr,6*\yr)[]{$T$};

    \cone{9}{$1$}{$\cdots$}{$i$}{$\cdots$}{$\ell$};
    \node at (13.5*\xr,9*\yr)[]{\txtdec$T-i$};
    \conex{10}{\txtdec$T-1$}{\txtdec$\cdots$}{\txtdec$T-i$}{\txtdec$\cdots$}{\txtdec$T-\ell$}{0.675};
    \node at (13.5*\xr,10*\yr)[]{$i$};

    \ctwo{11}{$1$}{$\cdots$}{$j$}{$\cdots$}{$\ell$};
    \node at (13.5*\xr,11*\yr)[]{\txtdec$T-j$};
    \ctwox{12}{\txtdec$T-1$}{\txtdec$\cdots$}{\txtdec$T-j$}{\txtdec$\cdots$}{\txtdec$T-\ell$}{0.675};
    \node at (13.5*\xr,12*\yr)[]{$j$};

    \rec{8.5}{1}{2.5}{0}{$0$};
    \rec{3}{2.5}{2}{0.5}{$0$};
    \rec{11}{2}{0}{0}{$0$};
    \rec{8}{3}{2}{0}{$0$};
    \rec{13.5}{2}{1.5}{1}{$0$};

    \rec{6}{6}{5}{2}{$0$};
    \rec{13.5}{7.5}{0}{0.5}{$0$};

    \rec{8.5}{9}{2.5}{0}{$0$};\rec{8.5}{10}{2.5}{0}{$0$};
    \rec{3}{11}{2}{0}{$0$};\rec{11}{11}{0}{0}{$0$};
    \rec{3}{12}{2}{0}{$0$};\rec{11}{12}{0}{0}{$0$};
    \rec{13.5}{13}{0}{0}{$0$};


    \draw [xbrace] (-1.4*\xr,3.25*\yr) -- (-1.4*\xr,0.75*\yr) node [black,midway,xshift=-0.8*\xr cm]{(1)};
    \draw [xbrace] (-1.4*\xr,8.25*\yr) -- (-1.4*\xr,3.75*\yr) node [black,midway,xshift=-0.8*\xr cm]{(2)};
    \draw [xbrace] (-1.4*\xr,13.25*\yr) -- (-1.4*\xr,8.75*\yr) node [black,midway,xshift=-0.8*\xr cm]{(3)};

    \draw [xbrace] (0.75*\xr,-0.5*\yr) -- (11.25*\xr,-0.5*\yr) node [black,midway,yshift=-0.8*\yr cm]{$\triangleq V(G)$};
    \draw [xbrace] (11.75*\xr,-0.5*\yr) -- (\xmx*\xr,-0.5*\yr) node [black,midway,yshift=-0.8*\yr cm]{$\triangleq E(G)$};

    \end{tikzpicture}
  \caption{Illustration of the instance obtained in the proof of~\cref{thm:fair_knapsack_wone}.
  Herein, $n_b^c$ denotes vertex~$b$ in color class~$c$, where each color class contains~$\ell$ vertices.
  In the presented example, the vertices~$n_i^1$ and~$n_j^2$ are adjacent.
  Blocks containing a zero indicate that the corresponding entries are zero.
  }
  \label{fig:fair_knapsack_wone}
\end{figure}
Let $T = |V(G)|$. We set the set of items to $V(G) \cup E(G)$, that is we associate one item with each vertex and with each edge. We construct the set of voters as follows (unless specified otherwise, by default we assume that a voter assigns utility of zero to an item):
\begin{enumerate}
\item For each color we introduce one voter who assigns utility of $T$ to each vertex with this color. Clearly, there are~$k$~such voters.
\item For each pair of two different colors we introduce $k-2$ voters, each assigning utility of $T$ to each edge that connects two vertices with these two colors. There are $(k-2){k \choose 2}$ such voters.
\item For each ordered pair of colors, $c_1$ and $c_2$, with $c_1 \neq c_2$ we introduce two voters, call them $a$ and $b$, with the following utilities. 
Consider the set of vertices with color~$c_1$ and rename them in an arbitrary way so that they can be put in a sequence $n_1, n_2, \ldots, n_{\ell}$. 
For each $i \in [\ell]$ voter~$a$ assigns utility~$i$ to vertex~$n_i$ and utility~$(T - i)$ to each edge that connects~$n_i$ with a vertex with color~$c_2$. Voter~$b$ assigns utility~$(T-i)$ to~$n_i$ and utility~$i$ to each edge that connects~$n_i$ with a vertex with color~$c_2$. There are $2k(k-1)$ such voters.
\end{enumerate}
We set the cost of each item to one, and the total budget to $B = k + {k \choose 2}$. By a simple calculation one can check that the total number of voters is equal to $k + (k-2)\cdot {k \choose 2} + 2k(k-1) = kB$. This completes our construction.

First, observe that in total each item is assigned utility~$kT$ from all the voters. 
Indeed, each item corresponding to a vertex gets utility of $T$ from exactly one voter from the first group, and total utility of $(k-1) \cdot T$ from $2(k-1)$ voters from the third group. 
Similarly, each item corresponding to an edge gets utility of $T$ from $k-2$ voters from the second group, and total utility of $2 \cdot T$ from four voters from the third group. Thus, independently of how we select $B$ items, the sum of the utilities they are assigned from the voters will always be the same, that is $BkT$. 
Thus, clearly the Nash welfare would be maximized if the total utility assigned to the selected items by each voter is the same, and equal to~$T$. 
Only in such case the Nash welfare would be equal to $(T+1)^{kB}$. 
We will show, however, that each voter assigns to the set of $B$ items utility~$T$ if and only if~$k$ out of such items are vertices with~$k$ different colors, the remaining~${k \choose 2}$ of such items are edges, and each selected edge connects two selected vertices.

Indeed, it is easy to see that if the selected set of items has the structure as described above, then each voter assigns to this set the utility of $T$. We will now prove the other implication. Assume that for the set of $B$ items $S$ each voter assigns total utility of $T$. By looking at the first group of voters, we infer that $k$ items from $S$ correspond to the vertices, and that these $k$ vertices have different colors. By looking at the second group of voters, we infer that for each pair of two different colors, $S$ contains exactly one edge connecting vertices with such colors. 
Finally, by looking at the third group of voters we infer that each edge from~$S$ that connects colors~$c_1$ and~$c_2$ is adjacent to the vertices from~$S$ with colors~$c_1$ and~$c_2$. This completes the proof.       
\end{proof}

By~\cref{thm:fair_knapsack_wone} we presumably cannot hope for an FPT algorithm for~$\fknap{}$ when parameterized by the number~$n$ of voters. 
However, each instance~$I$ of~\fknap{} with unarily encoded utilities is solvable in~$O(|I|^{f(n)})$ time (that is, it is in~$\xp$ when parameterized by~$n$), where~$f$ is some computable function only depending on~$n$.

\begin{theorem}\label{thm:fair_number_of_voters_xp}
For unarily encoded utilities, \fknap{} is in $\xp$ when parameterized by the number of voters.
\end{theorem}

\begin{proof}
We provide an algorithm based on dynamic programing. 
We construct a table~$T$ where for each sequence of $n+1$ integers, $z_1, z_2, \ldots, z_n$, and $i$, entry $T[z_1, z_2, \ldots, z_n, i]$ represents the lowest possible value of the budget~$x$ such that there exists a knapsack~$S$ with the following properties:
\begin{inparaenum}[(i)]
\item the total cost of all items in the knapsack is equal to $x$ (i.e., $x = \sum_{a \in S}c(a)$),
\item the last index of an item in the knapsack $S$ is $i$ (i.e., $i = \max_{a_j \in S}j$), and
\item for each voter $v_j$ we have that $\sum_{a \in S}u_j(a) = z_j$.  
\end{inparaenum}
This table can be constructed recursively:
\begin{align*}
T[z_1, z_2, \ldots, z_n, i] = 
 c(a_i) + \min_{j \leq i} T[z_1-u_1(a_i),  \ldots, z_n - u_n(a_i), j] \text{.}
\end{align*}
We handle the corner cases by setting $T[0, 0, \ldots, 0, i] = 0$ for each~$i$, and $T[z_1, z_2, \ldots, z_n, i] = \infty$ whenever $z_i < 0$ for some $i \in [n]$. 

Clearly, if $n$ is fixed and if the utilities are represented in unary encoding, then the table can be filled in polynomial time. Now, it is sufficient to traverse the table and to find the entry $T[z_1, z_2, \ldots, z_n, i] \leq B$ which maximizes $\prod_{j = 1}^n (z_j+1)$. 
\end{proof}

On the positive side, with stronger requirements on the voters' utilities, that is, if the number of different values over the utilities is small, we can strengthen \cref{thm:fair_number_of_voters_xp} and prove membership in $\fpt$ (using integer linear programming).

\begin{theorem}
\fknap{} is $\fpt$ when parameterized by the combination of the number of voters and the number of different values that a utility function can take. 
\end{theorem}

\begin{proof}
We will use the classic result of Lenstra~\cite{len:j:integer-fixed} which says that an integer linear program (ILP) can be solved in $\fpt$ time with respect to the number of integer variables. We will also use a recent result of Bredereck~et~al.~\cite{mixedIntegerProgrammingFPT} who proved that one can apply concave/convex transformations of certain variables in an ILP, and that such a modified program can be still solved in an $\fpt$ time. We construct an ILP as follows. 
Let $U$ be the set of values that a utility function can take. For each vector $z = (z_1, \ldots, z_n)$ with $z_i \in U$ for each $i$, we define $A_z$ as the set of items~$a$ such that for each voter~$v_i$ we have $u_i(a) = z_i$. 
Intuitively, $A_z$ describes a subcollection of the items with the same ``type'': such items are indistinguishable when we look only at the utilities assigned by the voters; they may vary only with their costs. For each such a set $A_z$ we introduce an integer variable $x_z$ which intuitively denotes the number of items from the optimal solution that belong to $A_z$. Further, we construct a function $f_z$ such that $f_z(x)$ is the cost of the $x$ cheapest items from $A_z$; clearly $f_z$ is convex. We formulate the following program:   

\begin{alignat*}{3}
 & \text{maximize:} && \sum_{v_i \in V}\log\left(\sum_{z \in U^n}z_i \cdot x_z\right) & \\
 & \text{subject to:} \quad && \sum_{z \in U^n}f_z(x_z) \leq B &\\
                 &          &&  x_{z} \in \integers, \quad & z \in U^n 
\end{alignat*}
The above program uses concave transformations (logarithms) for the maximized expression, and convex transformations (functions~$f_z$) in the left-hand sides of the constraints, so we can use the result of Bredereck~et~al.~\cite{mixedIntegerProgrammingFPT} and claim that this program can be solved in an $\fpt$ time with respect to the number of integer variables. This completes the proof.
\end{proof}

\subsection{Fair Knapsack under Restricted Domains}

In contrast to~\ibknap{} and~\dknap{}, both being solvable in polynomial time on restricted domains, \fknap{} remains $\np$-hard on utility profiles that are even both, single-peaked and single-crossing.

\begin{theorem}
  \label{thm:fair_knapsack_hard_sp}
 \fknap{} is $\np$-hard even on single-peaked single-crossing domains, when the costs of all items are equal to one, and the utilities of each voter come from the set $\{0, \ldots, 6\}$.
\end{theorem}

\begin{proof}
 We give a many-one reduction from the $\np$-hard \textsc{Exact-Set-Cover (X3C)} problem: 
 Given a universe~$U$ with~$3k$ elements and a set~$\calF$ of $3$-sized subsets of~$U$, the question is to decide whether there exist exactly~$k$ subsets in~$\calF$ that cover~$U$. 
 Without loss of generality, we can additionally assume that each element in~$U$ appears in exactly three sets from~$\calF$.
 Given an instance $(U=\{e_1,\ldots,e_n\},\calF=\{F_1,\ldots,F_m\})$ of \textsc{X3C} (note that~$n=m=3k$), we compute an instance of the problem of computing a fair knapsack as follows (the utilities of the voters are depicted in \cref{fig:fk-spsc}). 
 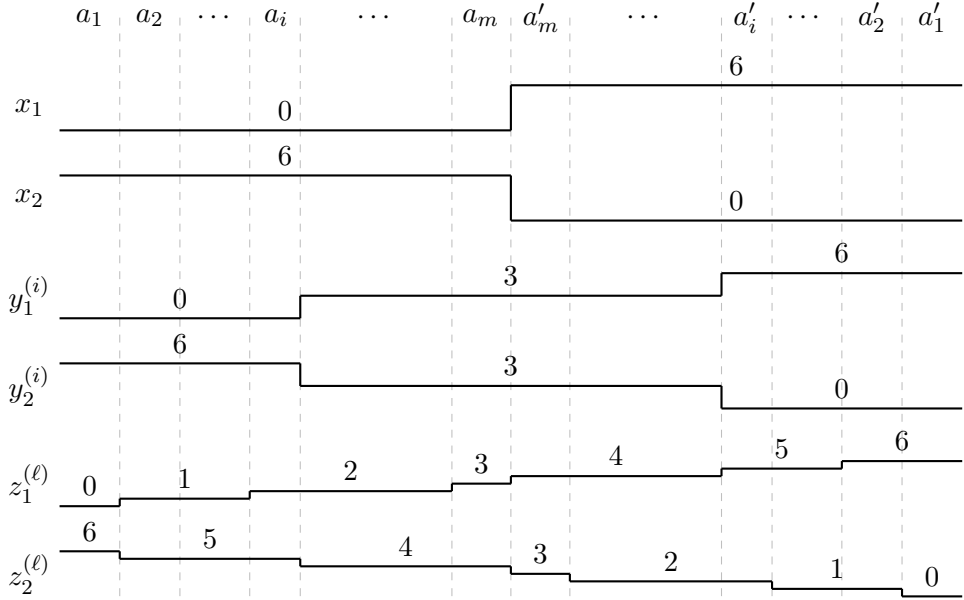
\begin{figure}[t]
 \centering
 \begin{tikzpicture}
    \tikzstyle{xline}=[thick];
    \def\xres{0.8}
    \def\yres{0.1}
    \def\mx{7}
    \def\tmx{2*\mx}
    \def\xs{\tmx*\xres}
    \def\ys{12*\yres}
    \def\yss{25*\yres}
    \def\ymx{15}
    \draw[-, dashed, lightgray] (0.5*\xres,\ymx*\yres) -- (0.5*\xres,-2.5*\yss);
    \draw[-, dashed, lightgray] (1.5*\xres,\ymx*\yres) -- (1.5*\xres,-2.5*\yss);
    \draw[-, dashed, lightgray] (0.38*\mx*\xres,\ymx*\yres) -- (0.38*\mx*\xres,-2.5*\yss);
    \draw[-, dashed, lightgray] (0.5*\mx*\xres,\ymx*\yres) -- (0.5*\mx*\xres,-2.5*\yss);
    \draw[-, dashed, lightgray] (0.86*\mx*\xres,\ymx*\yres) -- (0.86*\mx*\xres,-2.5*\yss);
    \draw[-, dashed, lightgray] (1*\mx*\xres,\ymx*\yres) -- (1*\mx*\xres,-2.5*\yss);
    \draw[-, dashed, lightgray] (\xs-0.5*\xres,\ymx*\yres) -- (\xs-0.5*\xres,-2.5*\yss);
    \draw[-, dashed, lightgray] (\xs-1.5*\xres,\ymx*\yres) -- (\xs-1.5*\xres,-2.5*\yss);
    \draw[-, dashed, lightgray] (\xs-0.38*\mx*\xres,\ymx*\yres) -- (\xs-0.38*\mx*\xres,-2.5*\yss);
    \draw[-, dashed, lightgray] (\xs-0.5*\mx*\xres,\ymx*\yres) -- (\xs-0.5*\mx*\xres,-2.5*\yss);
    \draw[-, dashed, lightgray] (\xs-0.86*\mx*\xres,\ymx*\yres) -- (\xs-0.86*\mx*\xres,-2.5*\yss);

    \node at (0*\xres,\ymx*\yres)[]{$a_1$};
    \node at (1*\xres,\ymx*\yres)[]{$a_2$};
    \node at (0.30*\mx*\xres,\ymx*\yres)[]{$\cdots$};
    \node at (0.5*\mx*\xres-0.4*\xres,\ymx*\yres)[]{$a_{i}$};
    \node at (0.68*\mx*\xres,\ymx*\yres)[]{$\cdots$};
    \node at (\mx*\xres-0.5*\xres,\ymx*\yres)[]{$a_{m}$};
    \node at (\xs-\mx*\xres+0.5*\xres,\ymx*\yres)[]{$a_m'$};
    \node at (\xs-0.68*\mx*\xres,\ymx*\yres)[]{$\cdots$};
    \node at (\xs-0.5*\mx*\xres+0.4*\xres,\ymx*\yres)[]{$a_i'$};
    \node at (\xs-0.30*\mx*\xres,\ymx*\yres)[]{$\cdots$};
    \node at (\xs-1*\xres,\ymx*\yres)[]{$a_2'$};
    \node at (\xs-0*\xres,\ymx*\yres)[]{$a_1'$};
%

    \node at (0-\xres,3*\yres)[]{$x_1$};
    \draw[xline] (0-0.5*\xres,0) to node[above,midway]{0}(\mx*\xres,0); 
    \draw[xline] (\mx*\xres,0) --  (\mx*\xres,6*\yres);
    \draw[xline] (\mx*\xres,6*\yres) to node[above,midway]{6}(\tmx*\xres+0.5*\xres,6*\yres);

    \node at (0-\xres,3*\yres-\ys)[]{$x_2$};
    \draw[xline] (\xs-0+0.5*\xres,0-\ys) to node[above,midway]{0}(\xs-\mx*\xres,0-\ys); 
    \draw[xline] (\xs-\mx*\xres,0-\ys) --  (\xs-\mx*\xres,6*\yres-\ys);
    \draw[xline] (\xs-\mx*\xres,6*\yres-\ys) to node[above,midway]{6}(\xs-\tmx*\xres-0.5*\xres,6*\yres-\ys);

    \node at (0-\xres,3*\yres-\yss)[]{$y_1^{(i)}$};
    \draw[xline] (0-0.5*\xres,0-\yss) to node[above,midway]{0}(0.5*\mx*\xres,0-\yss); 
    \draw[xline] (0.5*\mx*\xres,0-\yss) --  (0.5*\mx*\xres,3*\yres-\yss);
    \draw[xline] (0.5*\mx*\xres,3*\yres-\yss) to node[above,midway]{3}(1.5*\mx*\xres,3*\yres-\yss);
    \draw[xline] (1.5*\mx*\xres,3*\yres-\yss) -- (1.5*\mx*\xres,6*\yres-\yss);
    \draw[xline] (1.5*\mx*\xres,6*\yres-\yss) to node[above,midway]{6}(2*\mx*\xres+0.5*\xres,6*\yres-\yss);

    \node at (0-\xres,3*\yres-\ys-\yss)[]{$y_2^{(i)}$};
    \draw[xline] (\xs-0+0.5*\xres,0-\ys-\yss) to node[above,midway]{0}(\xs-0.5*\mx*\xres,0-\ys-\yss); 
    \draw[xline] (\xs-0.5*\mx*\xres,0-\ys-\yss) --  (\xs-0.5*\mx*\xres,3*\yres-\ys-\yss);
    \draw[xline] (\xs-0.5*\mx*\xres,3*\yres-\ys-\yss) to node[above,midway]{3}(\xs-1.5*\mx*\xres,3*\yres-\ys-\yss);
    \draw[xline] (\xs-1.5*\mx*\xres,3*\yres-\ys-\yss) -- (\xs-1.5*\mx*\xres,6*\yres-\ys-\yss);
    \draw[xline] (\xs-1.5*\mx*\xres,6*\yres-\ys-\yss) to node[above,midway]{6}(\xs-2*\mx*\xres-0.5*\xres,6*\yres-\ys-\yss);

    \node at (0-\xres,3*\yres-2*\yss)[]{$z_1^{(\ell)}$};
    \draw[xline] (0-0.5*\xres,0-2*\yss) to node[above,midway]{0}(0.5*\xres,0-2*\yss); 
    \draw[xline] (0.5*\xres,0-2*\yss) --  (0.5*\xres,\yres-2*\yss);
    \draw[xline] (0.5*\xres,\yres-2*\yss) to node[above,midway]{1}(0.38*\mx*\xres,\yres-2*\yss);
    \draw[xline] (0.38*\mx*\xres,\yres-2*\yss) -- (0.38*\mx*\xres,2*\yres-2*\yss);
    \draw[xline] (0.38*\mx*\xres,2*\yres-2*\yss) to node[above,midway]{2}(0.86*\mx*\xres,2*\yres-2*\yss);
    \draw[xline] (0.86*\mx*\xres,2*\yres-2*\yss) -- (0.86*\mx*\xres,3*\yres-2*\yss);
    \draw[xline] (0.86*\mx*\xres,3*\yres-2*\yss) to node[above,midway]{3}(1*\mx*\xres,3*\yres-2*\yss);
    \draw[xline] (1*\mx*\xres,3*\yres-2*\yss) -- (1*\mx*\xres,4*\yres-2*\yss);
    \draw[xline] (1*\mx*\xres,4*\yres-2*\yss) to node[above,midway]{4}(1.5*\mx*\xres,4*\yres-2*\yss);
    \draw[xline] (1.5*\mx*\xres,4*\yres-2*\yss) -- (1.5*\mx*\xres,5*\yres-2*\yss);
    \draw[xline] (1.5*\mx*\xres,5*\yres-2*\yss) to node[above,midway]{5}(2*\mx*\xres-1.5*\xres,5*\yres-2*\yss);
    \draw[xline] (2*\mx*\xres-1.5*\xres,5*\yres-2*\yss) --(2*\mx*\xres-1.5*\xres,6*\yres-2*\yss);
    \draw[xline] (2*\mx*\xres-1.5*\xres,6*\yres-2*\yss) to node[above,midway]{6}(2*\mx*\xres+0.5*\xres,6*\yres-2*\yss);

    \node at (0-\xres,3*\yres-\ys-2*\yss)[]{$z_2^{(\ell)}$};
    \draw[xline] (\xs-0+0.5*\xres,0-\ys-2*\yss) to node[above,midway]{0}(\xs-0.5*\xres,0-\ys-2*\yss); 
    \draw[xline] (\xs-0.5*\xres,0-\ys-2*\yss) --  (\xs-0.5*\xres,\yres-\ys-2*\yss);
    \draw[xline] (\xs-0.5*\xres,\yres-\ys-2*\yss) to node[above,midway]{1}(\xs-0.38*\mx*\xres,\yres-\ys-2*\yss);
    \draw[xline] (\xs-0.38*\mx*\xres,\yres-\ys-2*\yss) -- (\xs-0.38*\mx*\xres,2*\yres-\ys-2*\yss);
    \draw[xline] (\xs-0.38*\mx*\xres,2*\yres-\ys-2*\yss) to node[above,midway]{2}(\xs-0.86*\mx*\xres,2*\yres-\ys-2*\yss);
    \draw[xline] (\xs-0.86*\mx*\xres,2*\yres-\ys-2*\yss) -- (\xs-0.86*\mx*\xres,3*\yres-\ys-2*\yss);
    \draw[xline] (\xs-0.86*\mx*\xres,3*\yres-\ys-2*\yss) to node[above,midway]{3}(\xs-1*\mx*\xres,3*\yres-\ys-2*\yss);
    \draw[xline] (\xs-1*\mx*\xres,3*\yres-\ys-2*\yss) -- (\xs-1*\mx*\xres,4*\yres-\ys-2*\yss);
    \draw[xline] (\xs-1*\mx*\xres,4*\yres-\ys-2*\yss) to node[above,midway]{4}(\xs-1.5*\mx*\xres,4*\yres-\ys-2*\yss);
    \draw[xline] (\xs-1.5*\mx*\xres,4*\yres-\ys-2*\yss) -- (\xs-1.5*\mx*\xres,5*\yres-\ys-2*\yss);
    \draw[xline] (\xs-1.5*\mx*\xres,5*\yres-\ys-2*\yss) to node[above,midway]{5}(\xs-2*\mx*\xres+0.5*\xres,5*\yres-\ys-2*\yss);
    \draw[xline] (\xs-2*\mx*\xres+0.5*\xres,5*\yres-\ys-2*\yss) --(\xs-2*\mx*\xres+0.5*\xres,6*\yres-\ys-2*\yss);
    \draw[xline] (\xs-2*\mx*\xres+0.5*\xres,6*\yres-\ys-2*\yss) to node[above,midway]{6}(\xs-2*\mx*\xres-0.5*\xres,6*\yres-\ys-2*\yss);
  \end{tikzpicture}  
  \caption{Visualization of the utilities of the voters used in the proof of Theorem~9. %
  The solid lines can be interpreted as plots depicting the utilities of the voters from different items. For instance, agent $x_1$ assigns utility of $0$ to the items $a_1, \ldots a_m$, and utility of $6$ to the items $a_{m + 1}, \ldots a_{2m}$ (note that~$a_j'=a_{2m-j+1}$). Agents $z_1^{(\ell)}$ and $z_2^{(\ell)}$ depicted in the figure correspond to the element $e_{\ell}$ such that~$e_\ell\in F_2\cap F_i\cap F_m$.}
  \label{fig:fk-spsc}
 \end{figure}

 First, for each~$i\in[m]$, we introduce two items, $a_i$ and $a_{2m-i}$, that correspond to set~$F_i$, each with the cost of one.
 Further, we introduce three different types of voters:
 \begin{enumerate}[(1)]
  \item We add two voters~$x_1$ and~$x_2$ with $u_{x_1}(a_i)=u_{x_1}(a_{m+i})-6=0$ and $u_{x_2}(a_{m+i})=u_{x_2}(a_i)-6=0$ for all $i\in[m]$. 

  \item For each~$i\in [m]$, we add two voters~$y^{(i)}_1$ and $y^{(i)}_2$ with
  \[ u_{y^{(i)}_1}(a_j) = 
      \begin{cases} 
	 0, & 1\leq j\leq i, \\ 
	 3,& i<j<2m-i+1, \\ 
	 6,& 2m-i+1\leq j\leq 2m, 
      \end{cases} \]
  and $u_{y^{(i)}_2}(a_j)=u_{y^{(i)}_1}(a_{2m-j+1})$.
  \item For each~$i\in [n]$, we add two voters $z^{(i)}_1$ and $z^{(i)}_2$ with $u_{z^{(i)}_1}(a_j) = f_i(a_j)$ and~$u_{z^{(i)}_2}(a_j) = u_{z^{(i)}_1}(a_{2m-j+1})$, where for~$j \in [m]$
  \begin{align*}
    &f_i(a_j) = |\{\ell\mid 1\leq \ell\leq j, e_i\in F_\ell\}|,\\
    &f_i(a_{2m - j +1}) = 3 + |\{\ell\mid j\leq \ell\leq m, e_i\in F_\ell\}|.
  \end{align*}

 \end{enumerate}
 We set the budget to~$B=2k$ and the required Nash welfare to~$W=(6k+1)^{2+2m}(6k+2)^{2n}$.
 
 It is apparent that with the order~$(a_1,\ldots,a_{2m})$, this profile is single-peaked.
 For single-crossingness, note that the utilities of agents $x_i$, $y_i^{(j)}$, $z_{i}^{(j')}$ are increasing over~$(a_1,\ldots,a_{2m})$ if~$i=1$, and decreasing, if~$i=2$.
 Hence, the order of voters
 \begin{align*}
 (x_1,y_1^{(1)},\ldots,y_m^{(1)},z_1^{(1)},\ldots,z_m^{(1)},x_2,y_1^{(2)},&\ldots,y_m^{(2)},z_1^{(2)},\ldots,z_m^{(2)})
 \end{align*}
 witnesses single-crossingness.
 
 We will prove that $(U=\{e_1,\ldots,e_{3k}\},\calF=\{F_1,\ldots,F_m\})$ is a yes-instance for \textsc{X3C} if and only if the constructed instance of~\fknap{} is a yes-instance.
 
 ($\Rightarrow$)
 Let $\calF'=\{F_{b_1},\ldots,F_{b_k}\}\subseteq \calF$ be an exact cover of~$U$.
 We claim that $S=\{a_{b_i},a_{2m-b_i+1}\mid i\in [k]\}$ is a fair knapsack.
 First observe that $c(S)=2k\leq B$.
 We consider the welfare for each of the three types~(1)--(3) of voters separately.
 
 For~$x_1$ and~$x_2$ (1), we have $\sum_{a \in S} u_{x_1}(a)=\sum_{a \in S} u_{x_2}(a)=6k$.
 
 Next, consider the voters of type~(2).
 Consider~\mbox{$y^{(i)}_1$, $i\in[m]$}:
 \begin{align*}
    \sum_{j=1}^{k} [u_{y^{(i)}_1}(a_{b_j}) +  u_{y^{(i)}_1}(a_{2m-b_j+1})] &=\sum_{1\leq b_j\leq i} [u_{y^{(i)}_1}(a_{b_j}) +  u_{y^{(i)}_1}(a_{2m-b_j+1})] \\&\qquad+
    \sum_{ i<b_{j}\leq m} [u_{y^{(i)}_1}(a_{b_j}) +  u_{y^{(i)}_1}(a_{2m-b_j+1})] \\
    &= \big|\{j\mid 1\leq b_j\leq i\}\big| \cdot (0+6) + \big|\{j\mid i<b_{j}\leq m\}\big| \cdot (3+3) \\
    &= 6k \text{.}
 \end{align*}
 By symmetry, $\sum_{j=1}^{k} [u_{y^{(i)}_2}(a_{b_j}) +  u_{y^{(i)}_2}(a_{2m-b_j+1})] =6k$.
 
 Finally, consider the voters of type~(3).
 Consider a voter~$z^{(i)}_1$, $i\in[m]$.
 Let $j^*$ be the index such that~$e_i\in F_{b_{j^*}}$ (recall exact cover).
 We have
 \begin{align*}
   \sum_{a\in S}u_{z^{(i)}_1}(a) &= \sum_{j=1}^{k} [f_i(a_{b_j}) +  f_i(a_{2m-b_j+1})]  \\
    &= f_i(a_{b_{j^*}}) +  f_i(a_{2m-b_{j^*}+1}) + \sum_{j\in[k]\setminus \{j^*\}} [f_i(a_{b_j}) +  f_i(a_{2m-b_j+1})]  \\
    &= 6+1 + \sum_{j\in[k]\setminus \{j^*\}} \Big(|\{\ell\mid 1\leq \ell\leq b_j, e_i\in F_{\ell}\}| + 3 +|\{\ell\mid b_j\leq \ell\leq m, e_i\in F_{\ell}\}|\Big)  \\
    &= 6 + 1 + \sum_{j\in[k]\setminus \{j^*\}} \left(|\{\ell\mid 1\leq \ell\leq m, e_i\in F_{\ell}\}| + 3 \right) + \sum_{j\in[k]\setminus \{j^*\}} \big|\{e_i\} \cap F_{b_j}\big|  \\
    &= 6 + 1 + \sum_{j\in[k]\setminus \{j^*\}} 6 + \sum_{j\in[k]\setminus \{j^*\}} 0 \\
    &= 6k+1
 \end{align*}
 By symmetry, $\sum_{a\in S}u_{z^{(i)}_2}(a)=6k+1$.
 
 Hence, we get in total that the Nash welfare is equal to
 \begin{align*} 
  (1+6k)^{2}(1+6k)^{2m}(2+6k)^{2n}=W \text{.}
 \end{align*}
 
 ($\Leftarrow$) Let $S\subseteq \{a_1,\ldots,a_{2m}\}$ be a fair knapsack with $c(S)\leq 2k$ and with the Nash welfare at least equal to $W$.
 We will now show that the total utility that all voters assign to each item~$a_j \in A$ is equal to $6+6m+6n+3$. Indeed, the two voters from (1) assign to $a_j$ the total utility of $6$. Similarly, any pair of voters, $y^{(i)}_1$ and $y^{(i)}_2$, assigns utility of $6$ to $a_j$. Finally, observe that, whenever $e_i \in F_j$, then voters $z^{(i)}_1$ and $z^{(i)}_2$ assign utility of $7$ to $a_j$; otherwise they assign utility of $6$ to $a_j$. Since each set $F_j$ contains exactly 3 elements, we get that $a_j$ gets total utility of $(n-3)6 + 3\cdot 7$ from the voters from (3).

 Hence, $2k$ items contribute $2k(6+6m+6n+3)$ to the total utility, and so, for the Nash welfare to be equal to $W$, this total utility must be distributed as equally as possible among the voters. Specifically, $2m+2$ voters need to get the total utility of $6k$, and $2n$ voters must get the total utility of $6k+1$.

 Now, we claim that for each~$i\in[m]$, $a_i\in S\iff a_{2m-i+1}\in S$.
 Suppose this is not the case, and
 let $i\in[m]$ be the smallest index such that either (i)~$a_i\in S \land a_{2m-i+1}\not\in S$ or (ii) $a_i\not\in S \land a_{2m-i+1}\in S$.
 Consider the first case~(i).
 Let $k_1=|\{1\leq j\leq i\mid a_j\in S\}|$ and $k_2=|\{2m-i+1\leq j\leq 2m\mid a_j\in S\}|$.
 It holds that $k_1 \geq k_2+1$, and it follows for voter~$y^{(i)}_2$ that
 \begin{align*}
  \sum_{a_j \in S} u_{y^{(i)}_2}(a_j) &= k_1 6 + (2k-(k_1+k_2))\cdot 3 + k_2\cdot 0 \\
    &= 3(2k+(k_1-k_2)) \\
    &\geq 6k+3 \notin \{6k, 6k+1\}.
 \end{align*}
 Case~(ii) works analogously, and hence, our claim follows. 
 From this we infer that~$\sum_{a \in S} u_{y^{(\ell)}_i}(a)=6k$ for each~$i\in\{1,2\}$ and $\ell\in[m]$, and that $\sum_{a \in S} u_{x_i}(a)=6k$ for each~$i\in\{1,2\}$.
 Thus, for each voter $z$ from (3) it must be the case that $\sum_{a \in S} u_{z}(a)=6k+1$.
 
 Finally, we will prove that $\calF'=\{F_{b_1},\ldots,F_{b_k}\}=\{F_i\in F\mid a_i\in S,i\in[m]\}$ forms a cover of~$U$.
 Towards a contradiction suppose that there is an element~$e_i\in U$ such that $e_i$ is not covered by~$\calF'$.
 We consider voter~$z^{(i)}_1$.
 Observe that since~$e_i\not\in F_{b_j}$ for each $j\in[k]$, we have 
 \begin{align*}
  \sum_{j\in[k]} \big(u_{z^{(i)}_1}(a_{b_j}) + u_{z^{(i)}_1}(a_{2m-b_j+1})\big) &= \sum_{j\in[k]} 6 = 6k < 6k+1.
 \end{align*}
 Thus, we reached a contradiction, and consequently we get that every element in~$U$ is covered by~$\calF$. 
 This completes the proof.
\end{proof}

As we discussed in \cref{sec:defi}, if the voters' utilities come from the binary set $\{0,1\}$ and if the costs of the items are equal to one, then \fknap{} is equivalent to computing winners according to Proportional Approval Voting. 
For this case with single-peaked preferences, Peters~\cite{Pet17a} showed that the problem can be formulated as an integer linear program with totally unimodular constraints, and thus it is solvable in polynomial time. 
This makes our result interesting, as it shows that by allowing slightly more general utilities (coming from the set $\{0, \ldots, 6\}$ instead of $\{0, 1\}$) the problem becomes already $\np$-hard (even if we additionally assume single-crossingness of the preferences). 

\section{Conclusion}

We studied three variants of the knapsack problem in multiagent settings. 
One of these variants, selecting an individually best knapsack, has been considered in the literature before, and our work introduces the other two concepts: diverse and fair knapsack. 
Our paper establishes a relation between the multiagent knapsack model and a broad literature including work on multiwinner voting and on fair allocation. 
This way, we expose a variety of ways in which the preferences of the voters can be aggregated in different applications that are captured by the abstract model of the multiagent knapsack problem.    

Our complexity results are outlined in \Cref{tab:results}.
In summary, we showed that computing an individually~best or a diverse knapsack can be done efficiently under some constraints. 
On the contrary, we give multiple evidences that computing a fair knapsack is computationally hard. 
This also motivates the study of approximation and heuristic algorithms for computing a fair knapsack.   

\bibliographystyle{plainnat}
\bibliography{grypiotr2006} 

\end{document}